\documentclass[english]{amsart}
\usepackage{a4wide}
\usepackage[T1]{fontenc}   
\usepackage[latin1]{inputenc}
\usepackage[english]{babel} 
\usepackage{amsmath}
\usepackage{amsfonts}
\usepackage{amssymb}
\usepackage{mathrsfs}
\usepackage{url}
\usepackage{color}
\usepackage{enumerate}
\usepackage{hyperref}
\usepackage{graphicx}

\theoremstyle{plain}
\newtheorem{thm}{Theorem}

\newtheorem{lem}[thm]{Lemma}
{\bfseries}{\itshape}
{\bfseries}{\itshape}
\newtheorem{proposition}[thm]{Proposition}
\newtheorem{cor}[thm]{Corollary}
\newtheorem{corollary}[thm]{Corollary}
\newtheorem{lemma}[thm]{Lemma}

\theoremstyle{definition}
\newtheorem{definition}[thm]{Definition}

\theoremstyle{remark}
\newtheorem{remark}[thm]{Remark}
 \renewcommand{\leq}{\leqslant} 
\renewcommand{\geq}{\geqslant} 

\newcommand{\F}{\ensuremath{\mathbb{F}}}
\newcommand{\Fq}{\ensuremath{\mathbb{F}_q}}
\newcommand{\Fqm}{\ensuremath{\mathbb{F}_{q^m}}}
\newcommand{\fqm}{\Fqm}

\newcommand{\mat}[1]{\ensuremath{\boldsymbol{#1}}}
\newcommand{\code}[1]{\ensuremath{\mathscr{#1}}}

\newcommand{\AC}{\code{A}}
\newcommand{\BC}{\code{B}}
\newcommand{\CC}{\code{C}}
\newcommand{\DC}{\code{D}}

\newcommand{\GC}{\code{G}}

\newcommand{\Mm}{\mat{M}}

\newcommand{\av}{\mat{a}}
\newcommand{\bv}{\mat{b}}
\newcommand{\cv}{\mat{c}}
\newcommand{\dv}{\mat{d}}

\newcommand{\vv}{\mat{v}}
\newcommand{\xv}{{\mat{x}}}
\newcommand{\yv}{{\mat{y}}}

\newcommand{\GRS}[3]{\text{\bf GRS}_{#1}(#2,#3)}
\newcommand{\Alt}[3]{\code{A}_{#1}(#2, #3)}
\newcommand{\Goppa}[2]{\code{G}(#1, #2)}

\newcommand{\tr}{\operatorname{Tr}}
\newcommand{\Tr}{\tr}

\newcommand{\starp}[2]{{#1} \star {#2}}

\newcommand{\sq}[1]{#1^{\star 2}}
\newcommand{\sqb}[1]{\left(#1\right)^{\star 2}}

\newcommand{\eqdef}{\stackrel{\text{def}}{=}}

\newcommand{\Span}[2]{\left\langle \, #1 \, \right\rangle_{#2}}
\newcommand{\Fqspan}[1]{\left\langle \, #1 \, \right\rangle_{\Fq}}
\newcommand{\Fqmspan}[1]{\left\langle \, #1 \, \right\rangle_{\Fqm}}
\newcommand{\floor}[1]{\left\lfloor #1 \right\rfloor}
\newcommand{\ceil}[1]{\left\lceil #1 \right\rceil}

\newcommand{\ea}{{e_{\AC}}}
\newcommand{\eg}{{e_{\GC}}}
 
\title{On the dimension and structure of the square of the dual of a Goppa code}

\author{Rocco Mora}
\address{Inria, 2 rue Simone Iff, 75012 Paris, France\\
  Sorbonne Universit\'es, UPMC Univ Paris 06}
\email{rocco.mora@inria.fr}
\author{Jean--Pierre Tillich}
\address{Inria, 2 rue Simone Iff, 75012 Paris, France}
\email{jean-pierre.tillich@inria.fr}

\begin{document}
\maketitle
\begin{abstract}
    The Goppa Code Distinguishing (GD) problem asks to distinguish efficiently a generator matrix of a Goppa code from a randomly drawn one. We revisit a distinguisher for alternant and Goppa codes through a new approach, namely by studying the dimension of square codes. We provide here a rigorous upper bound for the dimension of the square of the dual of an alternant or Goppa code, while the previous approach only provided algebraic explanations based on heuristics.  Moreover, for Goppa codes, our proof extends to the non-binary case as well, thus providing an algebraic explanation for the distinguisher which was missing up to now. All the upper bounds are tight and match experimental evidence. Our work also introduces new algebraic results about products of trace codes in general and of dual of alternant and Goppa codes in particular, clarifying their square code structure. This might be of interest for cryptanalysis purposes. 
\end{abstract}

\section{Introduction}

\subsection*{The McEliece scheme.}
The McEliece encryption scheme \cite{M78}, which dates back to 1978, is a code-based cryptosystem built upon the family of binary Goppa codes. It is equipped with very fast encryption and decryption algorithms and has very small ciphertexts. It is also widely viewed as a viable quantum safe cryptosystem,  since the best quantum algorithm for breaking it \cite{KT17a} has exponential complexity and the corresponding exponent barely improves the exponent of the best classical algorithm \cite{BM17} by about $40$ percent.

Over the years, the attempts to attack McEliece scheme moved in two main directions. One hand, we have \textit{message-recovery attacks}. They consist in inverting the McEliece encryption without finding a trapdoor and make use of general decoding algorithms. Despite considerable improvements \cite{LB88,S88,CC98,MMT11,BJMM12,MO15,BM17}, all these algorithms have exponential complexity. The parameters of McEliece-like schemes have then been intentionally chosen to thwart this attack, because it is considered as the main threat to the scheme. Despite all these efforts, the original McEliece cryptosystem \cite{M78} based on binary Goppa codes remains, after more than forty years, unbroken, be it by a classical or a quantum computer. It is now the oldest public-key cryptosystem with this feature.

The other way to attack the cryptosystem is by seeking to recover the private key.
For a long time it was widely believed that even a simpler task which is just to distinguish efficiently a generator matrix of a Goppa code from a randomly drawn generator matrix with non negligible probability was unfeasible. This is the so called \textit{Goppa Code Distinguishing (GD) problem} as introduced by the authors of \cite{CFS01}. The nice feature of this problem is that it is possible to devise a security proof for the McEliece scheme based solely on the intractability of this problem and decoding a generic linear code \cite{S10}. 
The belief about GD problem hardness was basically justified by the fact that Goppa codes behave like random codes in many aspects. For instance, they asymptotically meet the Gilbert-Varshamov bound, their weight distribution is roughly the same as those of random codes and they generally have a trivial permutation group. The absence of significant breakthrough in key-recovery attacks also strengthened the idea that the Goppa Code distinguishing problem is difficult. This problem was actually used for a long time as a problem which basically captures the hardness of recovering the private key of a Goppa code.

\subsection*{A distinguisher for high rate.}

However, this belief was severely questioned in \cite{FGOPT11,FGOPT13} which gave a polynomial time algorithm that distinguishes between Goppa codes (or more generally alternant codes) and random ones from their generator matrices at least for very high rate codes.  It is based on the kernel of a linear system related to an algebraic system that encodes the key-recovery problem for McEliece cryptosystem instantiated with alternant or Goppa codes. Indeed, it was shown to have an unexpectedly high dimension. This distinguisher was later on given another interpretation in \cite{MP12}, where it was proved that this dimension is related to the dimension of the square of the dual of the public code. 
The algebraic explanations given in \cite{FGOPT13} do not represent however a rigorous proof of the dimension of the kernel sought, but they rely on heuristic considerations. Indeed, while a set of vectors is proposed as candidate for the kernel basis, its elements are neither proved to be independent nor a set of generators. 
Despite that the experiments run in \cite{FGOPT13} show a regular behavior when alternant codes are defined by picking at random support and multiplier vectors, it is possible to artificially choose alternant codes whose kernel dimensions are even larger than for random ones. Moreover, in the case of Goppa codes, even if a general formula for the dimension of the kernel was provided which  matches the experimental evidence, an algebraic explanation was only provided in the case of binary Goppa codes with square free Goppa polynomials. This explanation crucially relies on the fact that binary Goppa codes are in this case also Goppa codes of a higher degree (with a Goppa polynomial being the square of the original polynomial). Clearly, this approach does not generalize 
to non binary Goppa codes.

\subsection*{Our contribution}
In the present article, we revisit the distinguisher for random alternant codes and Goppa codes. We do so  by exploiting the link given by \cite{MP12}. Indeed we provide a rigorous upper bound on the dimension of the square code of the dual of an alternant or a Goppa code that coincides with the experiments. By using \cite{MP12}, this also gives a lower bound on the dimension of the kernel of the matrix considered in \cite{FGOPT13}. Together with results about the typical dimension of the square of random codes \cite{CCMZ15}, this provides the first rigorous analysis about the effectiveness of the approach pioneered in \cite{FGOPT11}, because the typical dimension of the square of a random code 
is way larger than this upper-bound on the dimension of the square of the dual of a Goppa or alternant code.

 Our approach relies on several new ingredients
\begin{itemize}
\item a new result about the square of trace codes showing that if essentially the square of a code is abnormally small then the square of its trace code is also abnormally small in a certain region of parameters. By interpreting the dual of an alternant code or a Goppa code as the trace of a generalized Reed-Solomon code (whose dimension of the square is known to be abnormally small \cite{W10}) this shows that the square of a dual of an alternant code or a Goppa code is abnormally small.
\item While this approach explains rigorously why alternant codes or Goppa codes can be distinguished for extremely large rates, lower rates require a much delicate analysis, in particular in the Goppa case. We do so, by noticing that the square of a trace of a code $\CC$ can be interpreted as a sum of traces of products of $\CC$ with $\CC^{q^i}$ (which denotes $i$ applications of the Frobenius map to $\CC$). In the case of  Goppa codes, we show that the traces of these products turn out to be duals of alternant codes of a remarkably low degree at least for small values of $i$ (see Theorem \ref{thm:goppa_e_qeg_1}). This accounts for the remarkably low dimension of the square of the dual of Goppa codes in all cases considered in \cite{FGOPT13}. 
\end{itemize} 
Interestingly enough,  the latter argument applies to {\em all} kinds of Goppa codes, be they binary or not and provides now not only a rigorous explanation of the distinguisher found in \cite{FGOPT11}, but also covers the non-binary Goppa code case as well. 
Note that even if this approach is not able to distinguish the Goppa codes proposed in the NIST competition as shown in Table \ref{table: comparison_rates}, because it only works for very high rate Goppa codes, this still raises the issue whether this distinguishing approach can be improved to lower the dimension of the Goppa codes that can be distinguished by this approach. The better understanding of the distinguisher obtained here might help to address this issue.

\begin{table}
\begin{tabular}{|c | c | c ||c | c|| c | c |} 
 \hline
 McEliece parameter  & $n$ & $m$ & $r$ in & $R$ in  & Largest  & Corresponding  \\ 
 set & & & \cite{BCLMNPPSSSW19} & \cite{BCLMNPPSSSW19} & distinguishable $r$ & $R$\\
 \hline\hline
 kem/mceliece348864 & 3488 & 12 & 64 & 0.77982 & 12 & 0.95872 \\ 
 \hline
 kem/mceliece460896 & 4608 & 13 & 96 & 0.72917 & 12 & 0.96615 \\ 
 \hline
 kem/mceliece6688128 & 6688 & 13 & 128 & 0.75120 & 15 & 0.97084 \\ 
 \hline
kem/mceliece6960119 & 6960 & 13 & 119 & 0.77773 & 16 & 0.97011 \\ 
 \hline
 kem/mceliece8192128 & 8192 & 13 & 128 & 0.79688 & 19 & 0.96985 \\
 \hline
 \end{tabular}
\caption{Comparison between Classic McEliece  and smallest distinguishable code rates.
This table provides a comparison between the parameters proposed for Classic McEliece and the largest order $r$ of a binary Goppa code that is distinguishable, with the corresponding relative rate $R$. 
We see in this table that the code rates suggested for Classic McEliece oscillate between 0.7 and 0.8 \cite[Section 3]{BCLMNPPSSSW19}, while for the same length $n$ and degree of the field extension $m$, the distinguisher works for rates closer to 1, meaning that the Goppa order $r$ must be smaller. }
 \label{table: comparison_rates}
 
\end{table}

\section{Notation and prerequisites} \label{section: notation}
In this section, we fix notation used throughout the article. We also recall basic definitions and well-known results about subfield subcodes and trace codes, component-wise products, square codes and some algebraic codes derived from generalized Reed-Solomon codes.

Let $\F$ be a generic finite field, $\Fq$ and $\Fqm$ the finite fields with $q$ and $q^m$ elements respectively, where $q$ denotes a prime power, and $m$ is a positive integer.  Given $\vv_1,\dots,\vv_k \in \F^n$, we denote with $\Span{\vv_1,\dots,\vv_k}{\F}$ the subspace of vectors in $\F^n$ spanned by $\{\vv_1,\dots,\vv_k\}$. An $[n,k]$-code over $\F$ is a linear subspace of vectors in $\F^n$ of dimension $k$. The positive integer $n$ is called the \textit{code length}. 

Vectors and matrices will be denoted by bold letters $\xv$, 
$\Mm$, and are capitalized for matrices. We will also use for a function $f$ acting on $\F$ and a vector $\xv=(x_i)_{1 \leq i \leq n}$ in $\F^n$ by $f(\xv)$ 
the vector $(f(x_i))_{1 \leq i \leq n}$.

\subsection{Subfield subcodes and trace codes}
In this article we will often consider a finite field $\Fqm$ and its subfield $\Fq$. It is therefore useful to fix a basis of $\Fqm$ over $\Fq$:
\[
\{\alpha_0,\dots,\alpha_{m-1}\}.
\] 
We will also make use of a \textit{normal basis}
\[
\{\beta,\beta^q,\dots,\beta^{q^{m-1}}\}
\]
whenever fruitful.

A useful linear map from $\Fqm$ to its subfield $\Fq$ is the trace operator.
\begin{definition}
Given the finite field extension $\Fqm/\Fq$, we define the \textit{trace} operator $\tr_{\Fqm/\Fq} \colon \Fqm \to \Fq$ for all $x\in \Fqm$ as
\[
\tr_{\Fqm/\Fq}(x)=\sum_{i=0}^{m-1} x^{q^i}.
\]
The definition extends to vectors $\xv\in\Fqm^n$ so that the trace acts component-wise:
\[
\tr_{\Fqm/\Fq}(\xv)=(\tr_{\Fqm/\Fq}(x_1),\dots,\tr_{\Fqm/\Fq}(x_n))
\]
and to codes $\CC$ over $\Fqm$
\[
\tr_{\Fqm/\Fq}(\CC)=\{\tr_{\Fqm/\Fq}(\cv) \mid \cv \in \CC\}.
\]
\end{definition}
We remark that if $\CC=\Fqmspan{\cv_i \mid 1\le i \le k}$ then $\tr_{\Fqm/\Fq}(\CC)$ is a linear code over $\Fq$ and
\[
\tr_{\Fqm/\Fq}(\CC)=\Fqspan{\tr_{\Fqm/\Fq}(\alpha_j \cv_i) \mid 0\le j<m, 1\le i\le k}.
\]
Informally, multiplying the generators of $\CC$ inside the trace by each element $\alpha_j$ of the extension field basis takes into account the fact that $\CC$ is a code over $\Fqm$, while $\tr_{\Fqm/\Fq}(\CC)$ over the subfield $\Fq$. So, if  $\dim_{\Fqm}\CC=k$ then typically $\dim_{\Fq} \tr_{\Fqm/\Fq}(\CC)=mk$, unless $\tr_{\Fqm/\Fq}(\CC)$ coincides with the ambient space.
From now on, we will omit the extension field $\Fqm/\Fq$ and simply write $\tr$, whenever the former is clear from the context.

The main classical result linking trace codes to subfield subcodes is Delsarte's theorem.
\begin{thm}{(Delsarte's theorem \cite{D75})} \label{thm: Delsarte}
Let $\CC$ be a code over $\Fqm$. Then
\[
(\CC_{\rvert \Fq})^\perp=\tr(\CC^\perp),
\]
where $\CC_{\rvert \Fq}\eqdef\CC \cap \Fq$ denotes the subfield subcode over $\Fq$ of $\CC$.
\end{thm}

\subsection{Reed-Solomon, alternant and Goppa codes}
We first recall the definitions of some well-known classes of algebraic codes. In the following we will always denote with $r$ the dimension of a generalized Reed-Solomon code. We start with the definition of the latter, seen as an \textit{evaluation code}: 
\begin{definition}
  Let $\xv=(x_1,\dots,x_n)\in\F^n$ be a vector of pairwise distinct entries and $\yv=(y_1,\dots,y_n)\in\F^n$ a vector of nonzero entries. The $[n, r]$ \textit{generalized Reed-Solomon (GRS) code} with \textit{support} $\xv$ and \textit{multiplier} $\yv$ is
  \[
  \GRS{r}{\xv}{\yv}\eqdef\{(y_1 P(x_1),\dots,y_n P(x_n)) \mid P \in \F[z], \deg P < r\}
  \]
  \end{definition}
  
  This code can be conveniently generated by vectors that are component-wise (also called Schur) products of $\xv$ and $\yv$. Recall that this product is defined as
 \begin{definition}
   The \textit{component-wise product} of two vectors $\av,\bv\in\F^n$ is defined as
   \[
   \starp{\av}{\bv}\eqdef(a_1 b_1,\dots,a_n b_n).
   \]
   \end{definition}

 Since any polynomial $P$ in $\F[z]$ of degree $<r$ can be written as a linear combination over $\F$ of powers of $z$ of degree $<r$, by using the notation given above we can see the GRS code as
  \[
  \GRS{r}{\xv}{\yv}\eqdef\Fqspan{\xv^a\yv \mid 0\le a<r},
  \]
  where $\xv^a$ stands for the component-wise product $\star$ of the vector $\xv$ with itself repeated $a$ times and where we omitted the  component-wise product symbol between $\xv$ and $\yv$.
  
  The dual of a GRS code is also a GRS code, where the support and the multiplier are related to the ones of the primal code. In order to explicit such relation we introduce the polynomial
  \[
  \pi_\xv(z)\eqdef \prod_{i=1}^n (z-x_i)\in\F[z].
  \]
  \begin{proposition} \cite[Theorem~4, p.~304]{MS86}\label{pr:dual_GRS} 
  Let $\GRS{r}{\xv}{\yv}$ be a GRS code of length $n$. Its dual is also a GRS code. In particular
  \[
  \GRS{r}{\xv}{\yv}^\perp=\GRS{n-r}{\xv}{\yv^\perp},
  \]
  where 
  \[
  \yv^\perp\eqdef\left(\frac{1}{\pi'_\xv(x_1)y_1},\dots,\frac{1}{\pi'_\xv(x_n)y_n}\right)
  \]
  and $\pi'_\xv$ is the derivative of $\pi_\xv$.
  \end{proposition}
  An $\textit{alternant code}$ can be defined as the subfield subcode of a GRS code:
  \begin{definition}
    Let $n\le q^m$, for some positive integer $m$. Let $\GRS{r}{\xv}{\yv}$ be the GRS code over $\Fqm$ of dimension $r$ with support $\xv \in \Fqm^n$ and multiplier $\yv\in (\Fqm^*)^n$. The \textit{alternant code} with support $\xv$ and multiplier $\yv$ and \textit{degree} $r$ over $\Fq$ is
    \[
    \Alt{r}{\xv}{\yv}\eqdef \GRS{r}{\xv}{\yv}^\perp \cap \F_q^n.
    \]
    The integer $m$ is called \textit{extension degree} of the alternant code.
  \end{definition}
  Note that by Proposition \ref{pr:dual_GRS} an alternant code is the subfield subcode of a GRS code:
  $$
   \Alt{r}{\xv}{\yv}\eqdef \GRS{n-r}{\xv}{\yv^\perp} \cap \F_q^n.
  $$
  We use the same notation as in \cite{MS86} and use the dimension $r$ and the multiplier $\yv$ of the dual GRS code, which turns out to be more convenient in our setting.  A well-known polynomial time decoding algorithm for the family of alternant codes allows to decode up to $\frac{r}{2}$ errors.
  We observe that for $m=1$ an alternant code is simply a GRS code. Therefore from now on we will always assume $m>1$.
  
  From Delsarte's theorem (Theorem~\ref{thm: Delsarte}) and by duality,
  \begin{align}
  \Alt{r}{\xv}{\yv}^\perp&=\left(\GRS{r}{\xv}{\yv}^\perp \cap \F_q^n\right)^\perp \nonumber\\
  &=\tr\left((\GRS{r}{\xv}{\yv}^\perp)^\perp\right) \nonumber\\
  &=\tr\left(\GRS{r}{\xv}{\yv}\right). \label{eq: dual_alt}
\end{align}  
The dimension of an alternant code of order $r$ built upon an extension field of degree $m$ has therefore dimension at least $n-rm$.
  There exists a subclass of alternant codes which is particularly attractive for cryptographic purposes:
  \begin{definition}
    Let $\xv\in\Fqm^n$ be a support vector and $\Gamma\in\Fqm[z]$ a polynomial of degree $r$ such that $\Gamma(x_i)\neq 0$ for all $i \in \{1,\dots,n\}$. The \textit{Goppa code} of degree $r$ with support $\xv$ and \textit{Goppa polynomial} $\Gamma$ is defined as
    \[
    \Goppa{\xv}{\Gamma}\eqdef\Alt{r}{\xv}{\yv},
    \]
    where $\yv\eqdef\left(\frac{1}{\Gamma(x_1)},\dots,\frac{1}{\Gamma(x_n)}\right).$
  \end{definition}
  The reason why binary Goppa codes are preferable to instantiate McEliece-like schemes is that, if the Goppa polynomial has no multiple roots, there exists a polynomial time algorithm to decode up to $r$ errors. This is a direct consequence of
  \begin{thm}\label{thm: binary_Goppa->Alt} \cite{P75}
    Let $\Goppa{\xv}{\Gamma}$ be a binary Goppa polynomial $\Gamma$ of degree $r$ and without multiple roots. Then
    \[\Goppa{\xv}{\Gamma}=\Goppa{\xv}{\Gamma^2}=\Alt{2r}{\xv}{\yv},\]
    where $y_i\eqdef\frac{1}{\Gamma(x_i)^2}$ for all $1\le i \le n$.
  \end{thm}
$\Alt{r}{\xv}{\yv}= \GRS{r}{\xv}{\yv}^\perp \cap \F_q^n.$

\subsection*{Square of codes}
GRS codes turn out to display a very peculiar property with respect to the component-wise/Schur product of codes which is defined from the component-wise/Schur product of vectors by
 \begin{definition}
The  \textit{component-wise product of codes} $\CC,\DC$ over $\F$ with the same length $n$ is defined as
   \[
      \starp{\CC}{\DC}\eqdef \Span{\starp{\cv}{\dv} \mid \cv \in \CC, \dv \in \DC}{\F}.       \]
     If $\CC=\DC$, we call $\sq{\CC}\eqdef\starp{\CC}{\CC}$ the \textit{square code} of $\CC$. 
 \end{definition}
It is easy to see that if $\CC=\Span{\cv_1,\dots,\cv_{k_1}}{\F}$ and  $\DC=\Span{\dv_1,\dots,\dv_{k_2}}{\F}$, a generating set for $\starp{\CC}{\DC}$ over $\F$ is given by a set of $k_1 k_2$ vectors,
\[\{\starp{\cv_i}{\dv_j} \mid 1\le i\le k_1, 1\le j \le k_2\}.
\]
However, when $\CC \cap \DC \neq \{\mathbf{0}\}$, some of the above elements are obviously redundant. In the extreme case, i.e. when $\CC = \DC$, the square code dimension is much smaller than $k^2$, where $k=\dim_{\F} \CC$. This is a consequence of the commutative property of the component-wise product: $\starp{c_i}{c_j}=\starp{c_j}{c_i}$. Thus we can give the following folklore result appearing for instance in \cite{CCMZ15}.
\begin{proposition}\label{prop: dim_sq}
 Let $\CC$ be a linear code  over $\F$ of dimension $k$ and length $n$. Then
\[ \dim_{\Fq} \sq{\CC} \le \min\left(n,\binom{k+1}{2}\right).\]
\end{proposition}

For a random linear code $\CC$ whose square does not fill the full space, the dimension of its square code is  $\binom{k+1}{2}$ with high probability, where $k$ is the dimension of $\CC$. However, there exist families of codes for which the inequality in Proposition~\ref{prop: dim_sq} is strict. This clearly translates into a distinguisher from random codes. Generalized Reed-Solomon codes represent an example of such behavior \cite{W10}. We namely have
\begin{proposition} \label{pr: square_GRS}
Let $\GRS{r}{\xv}{\yv}$ be a GRS code with support $\xv$, multiplier $\yv$ and dimension $k$. We have $\sq{\GRS{k}{\xv}{\yv}}=\GRS{2k-1}{\xv}{\yv^2}$.
Hence, if $k\le\frac{n+1}{2}$,
\[\dim_{\Fqm}\sq{(\GRS{r}{\xv}{\yv})}=2k-1.
\]
\end{proposition}
This follows immediately from the fact that the square of a $\GRS{r}{\xv}{\yv}$ can be written as
\begin{align*}
\sq{\GRS{k}{\xv}{\yv}}&=\Fqmspan{\starp{(\xv^a\yv)}{(\xv^b\yv)} \mid 0\le a,b < k}\\
&=\Fqmspan{\xv^{a+b}\yv^2 \mid 0\le a,b < k}\\
&=\Fqmspan{\xv^{c}\yv^2 \mid 0\le c < 2k-1}.
\end{align*}
Note that the square code dimension is here $2k-1$, i.e. it is linear in $k$ and not quadratic. Since the dual of a Reed-Solomon code is again a Reed-Solomon code, it turns out that this algebraic class is distinguishable for any rate. For other families, a square code-based distinguisher may occur only for certain rates. For instance Goppa codes (and more in general alternant codes) are distinguishable whenever the rate is high enough as we now recall. Again this is related to such square code considerations as we will explain.

\subsection*{The distinguisher of Goppa/alternant codes of \cite{FGOPT11,FGOPT13} and its relationship with square code considerations}

The dual of an alternant (or Goppa) code can also be distinguished from random codes when the primal code has a high enough rate, using the square code tool. The different behavior was already observed in \cite{FGOPT11}. Here however, the distinguisher was presented in terms of the kernel dimension $D$ of a linear system obtained by linearizing in the proper way the algebraic system that encodes the key-recovery problem for McEliece cryptosystem endowed with alternant or Goppa codes. Indeed let $\mat{P}=(p_{ij})_{i,j}$ be a generator matrix of an $[n,k]$ alternant (or Goppa) code $\CC$ in systematic form, i.e. with its first $k$ columns that form an identity block and consider the following linear system
\[
\mathscr{L}_p =\left\{\sum_{k+1\le j < j' \le n} p_{ij}p_{ij'} Z_{jj'}=0 \mid 1\le i \le k\right\}.
\]
The dimension $D$ of the solution space of this system turns out to be much smaller in the case of high rate Goppa or alternant codes than for random codes. A formula for $D$ coinciding with experimental evidence was given in \cite{FGOPT13} together with a convincing algebraic explanation for alternant and binary Goppa codes. 

It has been proved in \cite{MP12} that such $D$ is related to the dimension of the square of the dual code $\CC^\perp$. Indeed, from \cite[Proposition~1]{MP12},
\begin{equation}\label{eq:MP12}
\dim_{\F} \sqb{{\CC^\perp}}=\binom{\dim_{\F}({\CC^\perp}) +1}{2}-D.
\end{equation}
In terms of dimensions of the square codes, the formula for $D$ given in \cite{FGOPT13} together with \eqref{eq:MP12} predicts for a generic alternant code $\Fq$ of length $n$ and extension degree $m$ that
\begin{equation}
\label{eq:prediction_alternant}
\dim_{\Fq} \sq{(\Alt{r}{\xv}{\yv}^\perp)} = \min\left\{n, \binom{rm+1}{2}-\frac{m}{2}(r-1)\left((2e_{\AC}+1)r-2\frac{q^{e_{\AC}}-1}{q-1}\right)\right\},
\end{equation}
whereas for a generic Goppa code $\Goppa{\xv}{\Gamma}$  of length $n$ over $\Fq$ with Goppa polynomial $\Gamma(X) \in \Fqm[X]$ of degree $r$:
\begin{eqnarray} 
  \dim \sq{(\Goppa{\xv}{\Gamma}^\perp)}& = & \min\left\{n,\binom{rm+1}{2}-\frac{m}{2}(r-1)(r-2)\right\},\;\;\text{if $r < q-1$} \label{eq:prediction_Goppa_e=0}\\
  \dim \sq{(\Goppa{\xv}{\Gamma}^\perp)}& =&  \min \left\{n,\binom{rm+1}{2}-\frac{m}{2}r\left((2e_{\GC}+1)r-2(q-1)q^{e_{\GC}-1}-1\right)\right\}, \;\;\text{else,}
  \label{eq:prediction_Goppa_e>0}
  \end{eqnarray}
  where $\ea$ and $\eg$ are respectively defined by
  \begin{eqnarray*}
e_{\AC}&\eqdef &\max\{i \in \mathbb{N} \mid r\ge q^i+1\}=\floor{\log_q(r-1)}\\
e_{\GC}&\eqdef &\min \{i \in \mathbb{N} \mid r \le (q-1)^2q^{i}\} +1 = \ceil{\log_q\left(\frac{r}{(q-1)^2}\right)}+1.
\end{eqnarray*}

As shown in \cite{FGOPT13}, these formulas agree with extensive experimental evidence. 
 \section{A general result about the square of a trace code}\label{sec:general}

The dual of alternant codes and Goppa codes are trace codes of GRS codes. From Proposition \ref{pr: square_GRS} we know that square codes of GRS codes have an abnormally small dimension. A natural question is whether or not this implies that the 
square of the trace of a GRS code has itself a small dimension. More generally, this raises the following fundamental issue
of whether or not when the product of two codes $\CC$ and $\DC$ over $\fqm$ of length $n$  is smaller than $\min(n,\dim_{\Fqm} \CC \cdot\dim_{\Fqm} \DC)$ (which is the dimension we expect for random codes $\CC$ and $\DC$) then this property survives for trace codes, namely  do we have in this case
$$\dim_{\Fq} \starp{\Tr(\CC)}{\Tr(\DC)} < \min(n,\dim_{\Fq}  \Tr(\CC)  \cdot \dim_{\Fq} \Tr(\DC))?$$ This is related to open questions raised in \cite[C.4]{R15}. This is indeed the case up to some extent, due to the following proposition:

\begin{proposition}\label{pr:general}
Let $\CC$ and $\DC$ be two linear codes over $\Fqm$ with the same length $n$. Then
\[
\starp{\Tr(\CC)}{\Tr(\DC)}\subseteq\sum_{i=0}^{m-1}\Tr\left(\starp{\CC}{\DC^{q^i}}\right), \;\;\text{where $\DC^{q^i}\eqdef \{d^{q^i} \mid d \in \DC\}$.}
\]
\end{proposition}
\begin{proof}
It is readily verified that $\DC^{q^i}$ is a linear code over $\Fqm$. Let $c,d \in \Fqm$. We  have
\begin{align*}
    \tr(c)\cdot \tr(d)&=\left(\sum_{0 \le i\le m-1} c^{q^i}\right)\cdot \left(\sum_{0 \le i\le m-1} d^{q^i}\right) \\
    &=\sum_{\begin{smallmatrix}0\le i \le m-1\\
   0\le j\le m-1\end{smallmatrix}}c^{q^i}\cdot d^{q^j} \\
   &=\sum_{0 \le j \le m-1} \sum_{0\le i\le m-1} c^{q^i}\cdot d^{q^{(i+j \mod m)}} \\
   &=\sum_{0 \le j \le m-1} \sum_{0\le i\le m-1} c^{q^i}\cdot d^{q^{i+j}} \\
   &=\sum_{0 \le i\le m-1} \tr(c \cdot d^{q^i}).
\end{align*}
Because the trace acts component-wise on vectors, we also have for $\cv,\dv\in \Fq^n$,
\[
\tr(\cv) \star \tr(\dv)=\sum_{0 \le i\le m-1} \tr(\cv \star \dv^{q^i}).\]
Hence
\begin{align*}
    \tr(\CC) \star \tr(\DC)&=\Fqspan{\tr(\cv) \star \tr(\dv) \mid \cv \in \CC, \dv \in \DC}\\
    &=\Fqspan{\sum_{i=0}^{m-1} \tr\left(\starp{\cv}{\dv^{q^i}}\right) \mid \cv \in \CC, \dv \in \DC}\\
    &\subseteq\sum_{i=0}^{m-1} \Fqspan{\tr\left(\starp{\cv}{\dv^{q^i}}\right) \mid \cv \in \CC, \dv \in \DC}\\
    &=\sum_{i=0}^{m-1} \tr\left(\CC \star \DC^{q^i}\right).
\end{align*}
\end{proof}

Note that for an $\Fqm$-linear code $\CC$, $\dim_{\Fq} \Tr(\CC) \leq  \min(m \cdot \dim_{\Fqm} \CC,n)$, where $n$ is the code length of $\CC$ and $\DC$, and equality generally holds.
An easy corollary of this proposition is that 
\begin{corollary}\label{cor:general}
Let $\CC$ and $\DC$ be two $\Fqm$-linear codes of a same length and which are such that 
$\dim_{\Fq} \Tr(\CC) = m \cdot \dim_{\Fqm} \CC$ and $\dim_{\Fq} \Tr(\DC) = m \cdot \dim_{\Fqm} \DC$. We have
$$
\dim_{\Fq} \left( \Tr(\CC) \star \Tr(\DC) \right) - \dim_{\Fq}  \Tr(\CC)\cdot \dim_{\Fq} \Tr(\DC) \leq m \cdot \left( \dim_{\Fqm}( \starp{\CC}{\DC}) - \dim_{\Fqm} \CC\cdot  \dim_{\Fqm}\DC \right)
$$
\end{corollary}
\begin{proof}
We will drop in what follows the subscript indicating in the dimension if it  is taken by considering the corresponding code as an $\Fqm$ subspace or as an $\Fq$ subspace-- it will be clear from the context.
We have 
\begin{eqnarray*}
\dim \left( \Tr(\CC) \star \Tr(\DC) \right) & \leq & \sum_{i=0}^{m-1} \dim \Tr \left( \starp{\CC}{\DC^{q^i}}\right) \;\;\text{(by Prop. \ref{pr:general})}\\
& \leq & m \cdot \dim \left( \starp{\CC}{\DC}\right) +  \sum_{i=1}^{m-1} m \cdot \dim \left( \starp{\CC}{\DC^{q^i}}\right)\\
& \leq & m  \left( \dim( \starp{\CC}{\DC}) - \dim(\CC) \dim(\DC)\right) + m \cdot \dim \CC\cdot \dim \DC + m  \sum_{i=1}^{m-1}  \dim \CC \cdot \dim \left(\DC^{q^i}\right) \\
& \leq & m \left( \dim( \starp{\CC}{\DC}) - \dim \CC \cdot \dim \DC \right) + m^2 \dim \CC \cdot \dim \DC\\
& \leq & m \left( \dim( \starp{\CC}{\DC}) - \dim \CC \cdot \dim \DC \right) + \dim \Tr(\CC)\cdot \dim  \Tr(\DC).
\end{eqnarray*}
\end{proof}
\begin{remark}
In particular, this result implies that if we have two codes $\CC$ and $\DC$ over $\Fqm$ for which 
$\dim_{\Fqm} (\starp{\CC}{ \DC })< \dim_{\Fqm}  \CC  \cdot  \dim_{\Fqm} \DC$, then the same property survives for the corresponding trace codes:
$$
\dim_{\Fq} (\starp{\Tr(\CC)}{ \Tr(\DC) })< \dim_{\Fq}\Tr( \CC) \cdot  \dim_{\Fq} \Tr(\DC).
$$
\end{remark}

In the case $\CC=\DC$, namely if we consider square codes, Proposition \ref{pr:general} can be refined to give
\begin{proposition}\label{pr:square_general}
Let $\CC$  be a linear code over $\Fqm$. We have
\begin{eqnarray}
\tr\left(\starp{\CC}{\CC^{q^u}}\right)&=&\tr\left(\starp{\CC}{\CC^{q^{m-u}}}\right)\label{eq:symmetry}\\
\sqb{\Tr(\CC)} &\subseteq &\sum_{u=0}^{\lfloor m/2 \rfloor}\Tr\left(\starp{\CC}{\CC^{q^u}}\right) \label{eq:inclusion_general}\\
\dim_{\Fq} \left(\Tr\left(\starp{\CC}{\CC^{q^{m/2}}}\right)\right)& \leq &m \frac{(\dim_{\Fqm} (\CC))^2}{2} \;\; \text{if $m$ is even}
\label{eq:dimension}
\end{eqnarray}
\end{proposition}
\begin{proof}
{\em Proof of \eqref{eq:symmetry}.}
Let $\cv,\dv\in \CC$. Since the trace acts component-wise on vectors and $\tr(x)=\tr\left(x^{q^u}\right)$ for any $x\in \Fqm$ and natural number $u$, 
  \[
  \tr\left(\starp{\cv}{\dv^{q^u}}\right)=\tr\left((\starp{\cv}{\dv^{q^u}})^{q^{m-u}}\right)=\tr\left(\starp{\cv^{q^{m-u}}}{\dv^{q^m}}\right)=\tr\left(\starp{\dv}{\cv^{q^{m-u}}}\right)\in\tr\left(\starp{\CC}{\CC^{(q^{m-u})}}\right).
  \]
  This shows that $\tr\left(\starp{\CC}{\CC^{q^u}}\right)\subseteq\tr\left(\starp{\CC}{\CC^{q^{m-u}}}\right)$. By replacing $u$ by $m-u$ in the equality above, we obtain the reverse inclusion, which finishes the proof of \eqref{eq:symmetry}.
  
\noindent
{\em Proof of \eqref{eq:inclusion_general}.}  
\begin{eqnarray*}
\sqb{\Tr(\CC)} &\subseteq &\sum_{u=0}^{ m-1}\Tr\left(\starp{\CC}{\CC^{q^u}}\right) \;\;\text{(by Prop. \ref{pr:general} )}\\
& \subseteq & \sum_{u=0}^{\lfloor m/2 \rfloor}\Tr\left(\starp{\CC}{\CC^{q^u}}\right)\;\;\text{(by \eqref{eq:symmetry}).}
\end{eqnarray*}

\noindent
{\em Proof of \eqref{eq:dimension}.}
Let $\{\cv_1,\cdots,\cv_r\}$ be a basis of $\CC$ where $r \eqdef \dim (\CC)$. $\Tr\left(\starp{\CC}{\CC^{q^{m/2}}}\right)$ is generated by the $ \tr\left(\beta^{q^\ell} \starp{\cv_i}{\cv_j^{q^{m/2}}}\right)$ where $\{\beta, \beta^q,\cdots,\beta^{q^{m-1}}\}$ is a normal basis of $\Fqm$ and $\ell$ ranges over $\{0,\cdots, m-1\}$ and $i$, $j$ over $\{1,\cdots,r\}$. Since for any $0\le \ell\le\frac{m}{2}-1$, $1\le i,j \le r$,
\[
    \tr\left(\beta^{q^\ell} \starp{\cv_i}{\cv_j^{q^{m/2}}}\right) =\tr\left(\beta^{q^{\ell+\frac{m}{2}}} \starp{\cv_i^{q^{m/2}}}{\cv_j}\right)=\tr\left(\beta^{q^{\ell+\frac{m}{2}}} \starp{\cv_j}{\cv_i^{q^{m/2}}}\right),
\]
this implies that $ \tr\left(\CC*\CC^{q^\frac{m}{2}}\right)$ is generated by the (smaller) set $ \tr\left(\beta^{q^\ell} \starp{\cv_i}{\cv_j^{q^{m/2}}}\right)$ where $\ell$ ranges over $\{0,\cdots m/2-1\}$ and $i$, $j$ over $\{1,\cdots,r\}$. 
This is a set of cardinality $\frac{mr ^2}{2}$.
  \end{proof}
  
  Proposition \ref{pr:square_general} has a corollary which is similar to Corollary \ref{cor:general}, namely that
\begin{corollary}\label{cor:square_general}
Let $\CC$  be an $\Fqm$-linear code.
We have
\begin{equation}\label{eq:cor_dimension}
\dim_{\Fq} \sqb{\Tr(\CC)} \leq  m  \cdot \dim_{\Fqm}\sq{\CC} + \binom{m}{2} \left( \dim_{\Fqm} \CC \right)^2.
\end{equation} 
Furthermore if  
$\dim_{\Fq} \Tr(\CC) = m \cdot \dim_{\Fqm} \CC$ then
\begin{equation}\label{eq:cor_difference}
\dim_{\Fq} \sqb{\Tr(\CC)}  -\binom{\dim_{\Fq}  \Tr (\CC) +1}{2} \leq m \left[ \dim_{\Fqm}\sq{\CC} - \binom{\dim_{\Fqm} \CC +1}{2}\right].
\end{equation}
\end{corollary}
\begin{proof}
{\em Proof of \eqref{eq:cor_dimension}:}
by using \eqref{eq:inclusion_general} of Proposition \ref{pr:square_general}, we obtain
\begin{equation}
\label{eq:begin}
\dim \sqb{\Tr(\CC)}  \leq  \sum_{i=0}^{\lfloor \frac{m}{2} \rfloor} \dim \Tr \left( \starp{\CC}{\CC^{q^i}}\right).
\end{equation}
In the case of odd $m$, we deduce that 
\begin{eqnarray*}
\dim \sqb{\Tr(\CC)} & \leq & m \cdot \dim \sq{\CC} +  \sum_{i=1}^{\lfloor \frac{m}{2} \rfloor} m \cdot \dim \starp{\CC}{\CC^{q^i}} \\
& \leq & m  \cdot \dim \sq{\CC}  + m  \sum_{i=1}^{\lfloor \frac{m}{2} \rfloor}  \dim \CC \cdot \dim \CC^{q^i} \\
& \leq & m  \cdot \dim \sq{\CC}  + \frac{m(m-1)}{2} \left( \dim  \CC \right)^2 \\
& \leq & m  \cdot \dim \sq{\CC}  + \binom{m}{2} \left( \dim \CC \right)^2.
\end{eqnarray*}
On the other hand,  if $m$ is even, we have
\begin{eqnarray*}
\dim  \sqb{\Tr(\CC)} & \leq & m   \cdot \dim \sq{\CC}  +   \sum_{i=1}^{ \frac{m}{2} -1} m \cdot  \dim \CC \cdot \dim \CC^{q^i} +   \dim \Tr \left( \starp{\CC}{\CC^{q^{m/2}}}\right)\\
& \leq & m   \cdot \dim \sq{\CC} + \frac{m(m-2)}{2} \left( \dim \CC \right)^2 + \frac{m (\dim \CC)^2}{2} \;\;\text{(by using \eqref{eq:dimension} of Proposition \ref{pr:square_general})}\\
& \leq & m  \cdot \dim \sq{\CC}  + \binom{m}{2} \left( \dim \CC \right)^2.
\end{eqnarray*}
\noindent
{\em Proof of \eqref{eq:cor_difference}:}
\begin{eqnarray*}
\dim  \sqb{\Tr(\CC)}-  \binom{\dim  \Tr (\CC) +1}{2}& \leq & m  \cdot \dim \sq{\CC}  + \binom{m}{2} \left( \dim \CC \right)^2 - \binom{\dim  \Tr (\CC) +1}{2} \;\;\text{(by using \eqref{eq:cor_dimension})}\\
& \leq & m  \cdot \dim \sq{\CC}  + \frac{m(m-1)}{2} \left( \dim \CC \right)^2 -
\binom{m \dim \CC +1}{2}\\
& \leq & m  \cdot \dim \sq{\CC} + m \cdot \dim \CC \cdot \left[ \frac{m-1}{2} \dim \CC -\frac{m \dim \CC + 1}{2}\right]\\
& \leq & m  \cdot \dim \sq{\CC} - m \cdot \dim \CC \cdot \frac{\dim \CC + 1}{2}\\
& \leq & m \left[ \dim \sq{\CC} - \binom{\dim \CC +1}{2}\right].
\end{eqnarray*}
\end{proof}

Similarly to Corollary \ref{cor:general}, Corollary \ref{cor:square_general} implies that if the dimension of a square code $\sq{\CC}$ over $\Fqm$ is smaller than what we expect from a random code, namely that
$ \dim \left(\sq{\CC}\right) < \binom{\dim \CC +1}{2}$ (if $ \binom{\dim \CC +1}{2}$ is smaller than the code length) then this property survives for the trace code:
$$
\dim \sqb{\Tr(\CC)} < \binom{\dim  \Tr (\CC) +1}{2}.
$$


\section{Alternant case with $e_{\AC}=0$ and Goppa case with $\eg=0$} \label{section: e=0}
In this section, we are going to give a first upper bound on the dimension of the square of the dual of an alternant or Goppa code which is valid for all parameters and is tight when $e_{\AC}=0$ for random alternant codes and when $r< q-1$ for Goppa codes.
This will a direct application of the general results of Section \ref{sec:general} from which we derive:
\begin{thm}\label{thm:alternant_e=0}
Let $\Alt{r}{\xv}{\yv}$ be an alternant code over $\Fq$. Then
\begin{equation}\label{eq: dim_e=0} \dim_{\Fq} \sqb{\Alt{r}{\xv}{\yv}^\perp} \le \binom{rm+1}{2}-\frac{m}{2}(r-1)(r-2). 
\end{equation}
\end{thm}
\begin{remark}
Notice that the upper bound appearing in this theorem coincides exactly with the prediction of the dimension 
of the square code derived from the predictions given in \cite{FGOPT13} used together with \cite[Proposition~1]{MP12} 
in the alternant case corresponding to $\ea=0$, see \eqref{eq:prediction_alternant} or the Goppa case with $r<q-1$), see \eqref{eq:prediction_Goppa_e=0}. In other words, this theorem shows that this prediction is actually an upper bound on the dimension and the experimental evidence gathered in \cite{FGOPT13} actually shows that the dimensions of random alternant or Goppa codes agree with this upper bound in such a case.
\end{remark}

\begin{proof}
We let $\CC \eqdef \GRS{r}{\xv}{\yv}$. Note that $\Alt{r}{\xv}{\yv}^\perp = \tr (\CC)$. We apply Corollary \ref{cor:square_general} with such a $\CC$ and get that
\begin{eqnarray*}
\dim_{\Fq} \sqb{\tr(\CC)} & \leq & m \cdot \dim_{\Fqm} \sq{\CC} +\binom{m}{2}  \left(\dim_{\Fqm} \CC \right)^2 \\
& = & m (2r-1) + \frac{m(m-1)r^2}{2} \;\;\text{(by Proposition \ref{pr: square_GRS})}\\
& = & \left( 2(2r-1)+(m-1)r^2 \right) \frac{m}{2}\\
& = &  \left( r(mr+1)-(r-1)(r-2) \right) \frac{m}{2}\\
& = & \binom{rm+1}{2}-\frac{m}{2}(r-1)(r-2).
\end{eqnarray*}

\end{proof}
 
\section{Alternant case with $e_{\AC}>0$} \label{section: alternant_e>0}
In this section, we will show new linear relationships arising for alternant codes (hence also for Goppa codes) of high enough order $r$. More precisely, the threshold value for which new relations are guaranteed is $r\ge q+1$, i.e. $e_{\AC}>0$. 
Our main result in this section is that
\begin{thm}
\label{thm:alternant_e>0}
Let $\Alt{r}{\xv}{\yv}$ be an alternant code over $\Fq$. Then
\begin{equation}\label{eq:alternant_e>0} 
\dim_{\Fq} \sq{(\Alt{r}{\xv}{\yv}^\perp)} \le \binom{rm+1}{2}-\frac{m}{2}(r-1)\left((2\ea+1)r-2\frac{q^{\ea+1}-1}{q-1}\right).
\end{equation}
\end{thm}

\begin{remark}
Note that the upper bound on the dimension coincides with the prediction \eqref{eq:prediction_alternant} given for generic alternant codes. In other words, this theorem shows that this prediction is actually an upper bound on the dimension and the experimental evidence gathered in \cite{FGOPT13} actually shows that the dimensions of random alternant codes agree with this upper bound in such a case.
\end{remark}

The reason that we have a refinement of the upper bound of Theorem \ref{thm:alternant_e=0} for values of $r$ for which $\ea>0$ comes from the fact that  when we apply 
Proposition \ref{pr:square_general} with $\CC \eqdef \GRS{r}{\xv}{\yv}$ (which is the relevant quantity here since
$\tr(\CC)=\Alt{r}{\xv}{\yv}^\perp$) we get terms of the form $\tr\left(\starp{\CC}{\CC^{q^u}}\right)$ which will have a smaller dimension  than the generic upper bound $mr^2$. This is due to the fact that these  $\tr\left(\starp{\CC}{\CC^{q^u}}\right)$ will actually be duals of alternant codes for small values of $u$ as shown by the 
following lemma
\begin{lemma}
\label{lem:alternant}
Let $\CC \eqdef \GRS{r}{\xv}{\yv}$ and $f \eqdef \lfloor \log_q(r) \rfloor$. We have
\begin{eqnarray}
\tr\left(\starp{\CC}{\CC^{q^u}}\right)& \subseteq & \Alt{(r-1)(1+q^u)+1}{\xv}{\yv^{1+q^u}}^\perp\;\;\text{for all non-negative integers $u$,} \label{eq:alternant_inclusion}\\
\tr\left(\starp{\CC}{\CC^{q^u}}\right)& = & \Alt{(r-1)(1+q^u)+1}{\xv}{\yv^{1+q^u}}^\perp\;\text{for all integers $u$ in $\{0,\cdots,f\}$.}
\label{eq:alternant_equality} \end{eqnarray}
\end{lemma}

\begin{proof}
We first notice that 
 \[
  \tr\left(\starp{\CC}{\CC^{q^u}}\right)=\Fqspan{\tr\left(\alpha_l\xv^{a+bq^u}\yv^{1+q^u}\right) \mid  0\le l<m, 0\le a,b \leq  r-1}.
  \]
  Clearly the powers $a+bq^u$ are all smaller than or equal to $(r-1)(1+q^u)$. This directly implies 
  \eqref{eq:alternant_inclusion}.
  For the equality case, we observe that we get all the powers in $\{0,\cdots,(r-1)(1+q^u)\}$ as long as
  $r-1 \geq q^u-1$, that is $r \geq q^u$ which is equivalent to $u \leq \lfloor \log_q(r) \rfloor$.
\end{proof}
The previous lemma implies that
\begin{cor}
Let $\CC \eqdef  \GRS{r}{\xv}{\yv}$. For all non-negative integers $u$, we have
  $$ \dim_{\Fq} \tr\left(\starp{\CC}{\CC^{q^u}}\right) \le m((r-1)(q^u+1)+1).$$
\end{cor}
From this corollary we obtain directly Theorem \ref{thm:alternant_e>0}. The proof goes as follows.
\begin{proof}[Proof of Theorem \ref{thm:alternant_e>0}]
Let $\CC \eqdef  \GRS{r}{\xv}{\yv}$. 
\begin{align*}
\dim_{\Fq} \sq{(\Alt{r}{\xv}{\yv}^\perp)} &\le\dim_{\Fq} \sum_{u=0}^{\floor{\frac{m}{2}}}\tr\left(\starp{\CC}{\CC^{q^u}}\right) \;\;\text{(by Prop. \eqref{pr:square_general})}\\
&\le \sum_{u=0}^{\floor{\frac{m}{2}}} \dim_{\Fq} \tr\left(\starp{\CC}{\CC^{q^u}}\right) \\
&\le \sum_{u=0}^{e} m((r-1)(q^u+1)+1)+\left(\frac{m-1}{2}-e\right)mr^2\\
& \le \binom{rm+1}{2}-\frac{m}{2}(r-1)\left((2e+1)r-2\frac{q^{e+1}-1}{q-1}\right).
\end{align*}
Here, this inequality holds for any integer $e$ in $\{0,\cdots,\floor{\frac{m}{2}}\}$. To get the sharpest upper-bound, we 
 minimize the function
$\binom{rm+1}{2}-\frac{m}{2}(r-1)\left((2e+1)r-2\frac{q^{e+1}-1}{q-1}\right)$
 with respect to $e$. Obviously, this is equivalent to find the maximum over $\{0,\cdots,\floor{\frac{m}{2}}\}$ of the function
$$T(e)\eqdef er -\frac{q^{e+1}}{q-1}$$
in the discrete domain of non-negative integers.
We compute the discrete derivative $\Delta T : e \mapsto T(e+1)-T(e)$,
\begin{align*}
\Delta T(e)=T(e+1)-T(e)&=\left((e+1)r- \frac{q^{e+2}}{q-1}\right)-\left(er - \frac{q^{e+1}}{q-1} \right)\\&=r - q^{e+1}.
\end{align*}
This function is decreasing with $e$ and the smallest value of $e$ for which $\Delta T(e)  \leq 0$ corresponds to its maximum.
It is the smallest value of $e$ for which $r  \leq q^{e+1}$, i.e. $e=\floor{\log_q(r-1)}=\ea$.

\end{proof}
 \section{Goppa case with $r\ge q-1$}\label{section: Goppa_e>0}
In the previous two sections, we found linear relationships within the individual $\tr\left(\starp{\CC}{\CC^{q^u}}\right)$ subspaces, showing that they are spanned by less than $r^2m$ vectors if $r$ is large enough. We will see that the dimension of some $\tr\left(\starp{\CC}{\CC^{q^u}}\right)$ is even smaller in the Goppa case with $r\ge q-1$ (see \eqref{eq:inclusion} below). Moreover, they are no more disjoint, i.e. $\dim_{\Fq}\left(\tr\left(\starp{\CC}{\CC^{q^u}}\right)\cap \tr\left(\starp{\CC}{\CC^{q^v}}\right)\right) >0$ for some $0\le u<v\le \eg$. We will namely prove that

\begin{thm}\label{thm:goppa_e_qeg_1}
Let $\CC  \eqdef \GRS{r}{\xv}{\yv}$, where $y_i = \frac{1}{\Gamma(x_i)}$ and $\Gamma$ is a polynomial of degree $r$.
Let us define for any positive integer $v$ 
\begin{equation}
\BC_v \eqdef \Alt{r(q^v-q^{v-1}+1)}{\xv}{\yv^{q^v+1}}^\perp, \quad \text{and} \quad \BC_0 \eqdef \Alt{2r-1}{\xv}{\yv^2}^\perp.
\end{equation}
Let $f \eqdef \lfloor \log_q (r) \rfloor$. Then
\begin{eqnarray}
& &\tr\left(\starp{\CC}{\CC^{q^v}}\right) \subseteq \BC_v \quad \text{for all positive integers $v$,} \label{eq:inclusion}\\
& &\tr\left(\starp{\CC}{\CC^{q^u}}\right)=\BC_u \quad \text{for $0 \leq u \leq f$,} \label{eq:equality}\\
& &\tr\left(\starp{\CC}{\CC}\right)\subseteq \tr\left(\starp{\CC}{\CC^{q}}\right)\subseteq \dots\subseteq\tr\left(\starp{\CC}{\CC^{q^u}}\right) \quad \text{for $0 \leq u \leq f$.}
\label{eq:chain}
\end{eqnarray}
\end{thm}

Before proving this theorem it will be useful to state and prove two lemmas. The first one deals with the inclusion $\tr\left(\starp{\CC}{\CC^{q^v}}\right) \subseteq \BC_v$. For this, we first recall Lemma \ref{lem:alternant} which when applied to Goppa codes yields immediately  that for all integers $v$ we have:
\begin{eqnarray}
\tr\left(\starp{\CC}{\CC^{q^v}}\right) & \subseteq & \left\{ \tr\left( \frac{P(\xv)}{\Gamma(\xv)^{q^{v}+1}}\right)\mid \deg P < (r-1)(q^v+1)+1\right\} \label{eq:Goppa}\\
& = & \Alt{(r-1)(q^v+1)+1}{\xv}{\yv^{q^v+1}}^\perp. \nonumber
\end{eqnarray}
On the other hand, by definition of $\BC_v$ we have
\begin{equation}
\label{eq:BCv_definition}
\BC_v =  \left\{ \tr\left( \frac{P(\xv)}{\Gamma(\xv)^{q^{v}+1}}\right)\mid \deg P < r(q^v-q^{v-1}+1)\right\}.
\end{equation}
Depending on how $q$ compares to $r$, $r(q^v-q^{v-1}+1)$ might be either greater or smaller than $(r-1)(q^v+1)+1$. The key for showing the inclusion 
$\tr\left(\starp{\CC}{\CC^{q^v}}\right) \subseteq \BC_v$ will be to show that under certain conditions a codeword of the form $\tr\left( \frac{P(\xv)}{\Gamma(\xv)^{q^{v}+1}}\right)$ can be written as $\tr\left( \frac{Q(\xv)}{\Gamma(\xv)^{q^{v}+1}}\right) + \cv$ for some codeword $\cv \in \BC_v$ with a polynomial $Q$ such that 
$\deg Q < \deg P$. For performing this task we will use the following lemma.

\begin{lem}\label{lem:reduction}
Let $P(X)$ be a polynomial in $\Fqm[X]$ and consider the Euclidean division of $P$ by $\Gamma^{q^v-q^{v-1}+1}$:
$P = A \Gamma^{q^v-q^{v-1}+1}+ B$ where $\deg B < \deg \Gamma^{q^v-q^{v-1}+1}$. We have
\begin{eqnarray}
 \tr\left( \frac{P(\xv)}{\Gamma(\xv)^{q^v+1}}\right) &= &  \tr\left( \frac{A(\xv)^q\Gamma(\xv)}{\Gamma(\xv)^{q^v+1}}\right) + \cv, \;\text{for some $\cv \in \BC_v$}\label{eq:polynomial_switching}\\
 \deg A^q \Gamma & < & \deg P\quad \text{if $\deg P < r(q^v+1)$.} \label{eq:degree_reduction}
\end{eqnarray}
\end{lem}

\begin{proof}
For the first point, we just have to observe that 
\begin{eqnarray*}
\tr\left( \frac{P(\xv)}{\Gamma(\xv)^{q^v+1}}\right) & = &\tr\left( \frac{A(\xv)\Gamma(\xv)^{q^{v}-q^{v-1}+1}+B(\xv)}{\Gamma(\xv)^{q^v+1}}\right)  \\
& = & \tr\left( \frac{A(\xv)}{\Gamma(\xv)^{q^{v-1}}}\right) + \tr\left( \frac{B(\xv)}{\Gamma(\xv)^{q^{v}+1}}\right)\\
& = & \tr\left( \frac{A(\xv)^q}{\Gamma(\xv)^{q^{v}}}\right) + \tr\left( \frac{B(\xv)}{\Gamma(\xv)^{q^v+1}}\right)\quad \text{(since $\tr(\alpha^q) = \tr(\alpha)$ for $\alpha \in \Fqm$)}\\
& = &  \tr\left( \frac{A(\xv)^q \Gamma(\xv)}{\Gamma(\xv)^{q^{v}+1}}\right) + \cv \quad \quad \text{with $\cv \in \BC_v$.}
\end{eqnarray*}
The last equality follows immediately from the definition of $\BC_v$ (see \eqref{eq:BCv_definition}) and the fact that 
$\deg B < r(q^v-q^{v-1}+1)$. For the second point, let us just observe that either $A=0$  in which case \eqref{eq:degree_reduction} clearly holds or
$$
\deg A = \deg P - r(q^v-q^{v-1}+1)
$$
in which case
$$
\deg A^q \Gamma = q \deg P -qr(q^v-q^{v-1}+1)+r = q \deg P-r(q-1)(q^v+1).
$$
Clearly 
\begin{eqnarray*}
\deg A^q \Gamma < \deg P &\Leftrightarrow & q \deg P-r(q-1)(q^v+1) < \deg P \\
&\Leftrightarrow& \deg P < r(q^v+1).
\end{eqnarray*}
\end{proof}

The second lemma shows that the $\BC_v$'s form a nested family of codes:
\begin{lem}
\label{lem:nested}
We have for all $v \geq 0$:
$$
\BC_0 \subseteq \BC_1  \subseteq \cdots \subseteq \BC_v \subseteq \BC_{v+1} \subseteq \cdots
$$
\end{lem}
\begin{proof}
We will first prove that $\BC_v \subseteq \BC_{v+1}$ for $v \geq 1$.
Let us notice that 
\begin{align}
\BC_v&= \Fqspan{\tr\left(\alpha_l \xv^c \yv^{q^v+1}\right)\mid 0\le l<m, 0\le c< r(q^v-q^{v-1}+1)}\nonumber \\
& = \left\{ \tr\left( \frac{P(\xv)}{\Gamma(\xv)^{q^v+1}}\right): P \in \Fqm[X],\;\deg(P) < r(q^v-q^{v-1}+1)\right\} \label{eq:defBC}
\end{align}
Consider now a codeword $\tr\left( \frac{P(\xv)}{\Gamma(\xv)^{q^v+1}}\right)$ in $\BC_v$ and let us prove that it is 
in $\BC_{v+1}$. By \eqref{eq:defBC} we know that 
\begin{equation}
\label{eq:degP}
\deg P < r(q^v-q^{v-1}+1).
\end{equation} 
Then, let us notice that
\begin{eqnarray*}
\tr\left( \frac{P(\xv)}{\Gamma(\xv)^{q^v+1}}\right) &=  &\tr\left( \frac{P(\xv)\Gamma(\xv)^{q^{v+1}-q^v}}{\Gamma(\xv)^{q^{v+1}+1}}\right)\\
& = & \tr\left( \frac{A^q(\xv)\Gamma(\xv)}{\Gamma(\xv)^{q^{v+1}+1}}\right) +\cv
\end{eqnarray*}
where $\cv \in \BC_{v+1}$, $A$ is the quotient of the division of $P \Gamma^{q^{v+1}-q^v}$ by $\Gamma^{q^{v+1}-q^v+1}$, and we used  Lemma \ref{lem:reduction} (with $v+1$ instead of $v$ and $P \Gamma^{q^{v+1}-q^v}$ instead of $P$). Either $A=0$ and it follows immediately that 
$\tr\left( \frac{P(\xv)}{\Gamma(\xv)^{q^v+1}}\right)$ belongs to $\BC_{v+1}$, or else, if $v\ge 1$,
\begin{eqnarray*}
\deg A^q \Gamma & = & q \left(\deg P +\deg \Gamma (q^{v+1}-q^v) - \deg \Gamma (q^{v+1}-q^v+1) \right) + \deg \Gamma 
\nonumber \\
& < &q\left(r(q^v-q^{v-1}+1) +r (q^{v+1}-q^v)- r(q^{v+1}-q^v+1)\right] +r \;\; \text{(by definition of $\BC_v$ -see \eqref{eq:BCv_definition})} \\
& =& q\left( r(q^v-q^{v-1}) \right) +r \nonumber\\
& = & r (q^{v+1}-q^v+1). \label{eq:lastdeg}
\end{eqnarray*}

This shows that in this case too, we have  that $\tr\left( \frac{P(\xv)}{\Gamma(\xv)^{q^v+1}}\right)$ belongs to $\BC_{v+1}$.
If $v=0$ we proceed similarly, the only difference being that now the definition of $\BC_0$ is not obtained by  plugging $v=0$ into \eqref{eq:BCv_definition}, but is given by
\begin{align}
\BC_0
& = \left\{ \tr\left( \frac{P(\xv)}{\Gamma(\xv)^{2}}\right): P \in \Fqm[X],\;\deg(P) < 2r-1\right\}. \nonumber 
\end{align}

Consider a codeword $\tr\left( \frac{P(\xv)}{\Gamma(\xv)^{2}}\right)$ in $\BC_0$ and let us prove that it belongs to 
$\BC_1$. We proceed similarly as before and write
\begin{eqnarray*}
\tr\left( \frac{P(\xv)}{\Gamma(\xv)^{2}}\right) &=  &\tr\left( \frac{P(\xv)\Gamma(\xv)^{q-1}}{\Gamma(\xv)^{q+1}}\right).
\end{eqnarray*}
We can again apply Lemma \ref{lem:reduction} with $v=1$ and $P(\xv)\Gamma(\xv)^{q-1}$ (for the $P$ appearing in this lemma).
Notice that
\begin{eqnarray*}
\deg P(\xv)\Gamma(\xv)^{q-1} & < & 2r-1+r(q-1)\\
& < & r(q+1).
\end{eqnarray*}
This shows that we can write
$$
\tr\left( \frac{P(\xv)\Gamma(\xv)^{q-1}}{\Gamma(\xv)^{q+1}}\right)= \tr\left( \frac{Q(\xv)}{\Gamma(\xv)^{q^v+1}}\right) +\cv
$$
with $\cv$ belonging to $\BC_1$ and $\deg Q < \deg P\Gamma^{q+1}$. We can carry on this process further on $ \tr\left( \frac{Q(\xv)}{\Gamma(\xv)^{q^v+1}}\right)$ by reducing the degree of $Q$ with Lemma \ref{lem:reduction} until getting at the end a quotient $A$ of Lemma \ref{lem:reduction} equal to $0$, thus  showing that $\tr\left( \frac{P(\xv)}{\Gamma(\xv)^{2}}\right)$ belongs to $\BC_1$.
\end{proof}

We are now ready  to prove Theorem \ref{thm:goppa_e_qeg_1}.
\begin{proof}[Proof of Theorem \ref{thm:goppa_e_qeg_1}]
{\em Proof of \eqref{eq:inclusion}.}
We start by using Lemma \ref{lem:alternant} which implies that 
\begin{eqnarray}
\tr\left(\starp{\CC}{\CC^{q^v}}\right) & \subseteq & \left\{ \tr\left( \frac{P(\xv)}{\Gamma(\xv)^{q^{v}+1}}\right): \deg P < (r-1)(q^v+1)+1\right\} 
\nonumber \\
& = & \Alt{(r-1)(q^v+1)+1}{\xv}{\yv^{q^v+1}}^\perp \label{eq:Goppa_bis}. 
\end{eqnarray}
We consider a codeword $\tr\left( \frac{P(\xv)}{\Gamma(\xv)^{q^{v}+1}}\right)$ in $\tr\left(\starp{\CC}{\CC^{q^v}}\right)$ and apply Lemma \ref{lem:reduction} to it. Since $\deg P \leq (r-1)(q^v+1)<r(q^v+1)$ we can apply \eqref{eq:degree_reduction} and get that 
$$
\tr\left( \frac{P(\xv)}{\Gamma(\xv)^{q^{v}+1}}\right) = \tr\left( \frac{Q(\xv)}{\Gamma(\xv)^{q^{v}+1}}\right) +\cv,
$$
where $\cv \in \BC_v$ and $\deg Q < \deg P$. We can continue this process by applying iteratively Lemma \ref{lem:reduction} to $Q$ until the quotient $A$ becomes eventually the zero
polynomial, which shows that 
$$
\tr\left( \frac{P(\xv)}{\Gamma(\xv)^{q^{v}+1}}\right) \in \BC_v.
$$
\medskip
\noindent
{\em Proof of \eqref{eq:equality}.}
When $v \leq f$, Lemma \ref{lem:alternant} shows that the inclusion \eqref{eq:Goppa_bis} is actually an equality.
We therefore have
$$
\tr\left(\starp{\CC}{\CC^{q^v}}\right)=\Alt{(r-1)(q^v+1)+1}{\xv}{\yv^{q^v+1}}^\perp.
$$
On the other hand, from \eqref{eq:inclusion}, we know that 
\begin{equation}
\label{eq:reverse_inclusion}
\tr\left(\starp{\CC}{\CC^{q^v}}\right) \subseteq \BC_v = \Alt{r(q^v-q^{v-1}+1)}{\xv}{\yv^{q^v+1}}^\perp.
\end{equation}
Observe now that $r\geq q^v$ implies that $(r-1)(q^v+1)\geq r(q^v-q^{v-1}+1)-1$, since
\begin{eqnarray*}
(r-1)(q^v+1)-r(q^v-q^{v-1}+1)+1 & = & (r-q)q^{v-1}.
\end{eqnarray*}
This shows that 
$$
\BC_v = \Alt{r(q^v-q^{v-1}+1)}{\xv}{\yv^{q^v+1}}^\perp \subseteq \Alt{(r-1)(q^v+1)+1}{\xv}{\yv^{q^v+1}}^\perp
$$
which combined with the reverse inclusion \eqref{eq:reverse_inclusion} proves Equality \eqref{eq:equality}.

\medskip
\noindent
{\em Proof of \eqref{eq:chain}.} This is a direct consequence of Lemma \ref{lem:nested} and that the $\BC_v$'s coincide with the
$\tr\left(\starp{\CC}{\CC^{q^v}}\right)$'s in this range from \eqref{eq:equality}.
\end{proof}

A direct consequence of Theorem~\ref{thm:goppa_e_qeg_1} is that
\begin{cor}
Let $\Goppa{\xv}{\Gamma}$ be a Goppa code of order $r\ge q-1$ over $\Fq$. Then
\[ \sq{(\Goppa{\xv}{\Gamma}^\perp)}\subseteq \BC_e+\sum_{u=e+1}^{\floor{\frac{m}{2}}} \tr\left(\starp{\CC}{\CC^{q^u}}\right),\]
for any non-negative integer $e$.
\end{cor}
We can therefore conclude that
\begin{cor} \label{cor: goppa_dim}
Let $\Goppa{\xv}{\Gamma}$ be a Goppa code of order $r\ge q-1$ over $\Fq$. Then
\[ \dim_{\Fq} \sq{(\Goppa{\xv}{\Gamma}^\perp)} \le \binom{rm+1}{2}-\frac{m}{2}r\left((2\eg+1)r-2(q-1)q^{\eg-1}-1\right).\]
\begin{proof}
From Theorem~\ref{thm:goppa_e_qeg_1} and Lemma~\ref{lem:nested}, we have $\tr\left(\starp{\CC}{\CC^{q^u}}\right)\subseteq \BC_e$, for all $u\le e$. Therefore
\begin{align*}
\dim_{\Fq} \sq{(\Goppa{\xv}{\Gamma}^\perp)}&\le\dim_{\Fq}\left(\sum_{u=0}^{\floor{\frac{m}{2}}}\tr\left(\starp{\CC}{\CC^{q^u}}\right)\right)\\
&\le \dim_{\Fq}(\BC_e)+\sum_{u=e+1}^{\floor{\frac{m}{2}}} \dim_{\Fq} \tr\left(\starp{\CC}{\CC^{q^u}}\right) \;\;\text{(for arbitrary $e \in \{0,\cdots,\floor{\frac{m}{2}}\}$)}\\
&\le rm(q^e-q^{e-1}+1)+\left(\frac{m-1}{2}-e\right)mr^2\\
&= \binom{rm+1}{2}-\frac{m}{2}r\left((2e+1)r-2(q-1)q^{e-1}-1\right).
\end{align*}
We want now to minimize the function
$\binom{rm+1}{2}-\frac{m}{2}r\left((2e+1)r-2(q-1)q^{e-1}-1\right)$
 with respect to $e$. By removing the constant part in $e$, this clearly becomes equivalent to maximizing
\[T(e)\eqdef er-(q-1)q^{e-1}\]
over $\{0,\cdots,\floor{\frac{m}{2}}\}$.
We compute the discrete derivative $\Delta T : e \rightarrow T(e+1)-T(e)$,
$$
\Delta T(e)=T(e+1)-T(e)=((e+1)r-(q-1)q^{e+1-1})-(er-(q-1)q^{e-1})=r-(q-1)^2q^{e-1}.
$$
The maximum is attained at the least integer $e$ such that 
$\Delta T(e) \leq 0$. This corresponds to  the least integer $e$ such that $e \geq \log_q\left(\frac{r}{(q-1)^2}\right)+1$, 
i.e. $\ceil{\log_q\left(\frac{r}{(q-1)^2}\right)+1}=\eg$ minimizes the function $T$ and consequently
\[ \dim_{\Fq} \sq{(\Goppa{\xv}{\Gamma}^\perp)} \le \binom{rm+1}{2}-\frac{m}{2}r\left((2\eg+1)r-2(q-1)q^{\eg-1}-1\right).\]

\end{proof}
\end{cor}
\begin{remark}
From the computation given in the proof of Corollary~\ref{cor: goppa_dim} and knowing that the upper bound is usually an equality, we also infer that with high probability
\[
\BC_\eg = \sum_{u=0}^{\eg} \tr\left(\starp{\CC}{\CC^{q^u}}\right).\] This has been verified experimentally.

\end{remark}

\section*{Conclusion}
In this article we revisited the distinguisher for random alternant and Goppa codes presented for the first time in \cite{FGOPT13} through a different approach, namely using squares of codes. With this simple but powerful tool we were able to provide explicitly the linear relationships determining the distinguisher in a more straightforward way. We managed therefore to rigorously prove a tight upper bound for the dimension of the square of the dual of an alternant or Goppa code, while \cite{FGOPT13} only provides an algebraic explanation which does not however represent neither an upper or a lower bound. Our proof is also valid in the case of non-binary Goppa case, for which the conjectured distinguisher is only demonstrated experimentally in \cite{FGOPT13}. By doing this we got an unifying explanation for the behavior of all Goppa codes, which does not make use of specific features of the binary case. Finally, we illustrated an interesting property of the structure of the square of the dual of any Goppa code, relating it to the dual of another alternant code.

This connection could be of help for a potential key-recovery attack. Indeed a distinguisher can sometimes be turned into an attack. In the code-based cryptography setting, this is for instance the case for GRS codes. The uncommon dimension of the square of a GRS code leads to a successful key-recovery for several proposed variants of McEliece cryptosystem built upon this family of codes for any rate \cite{CGGOT14}. Despite the strong relation between generalized Reed-Solomon codes and alternant codes, the same attacks cannot be carried over from the former to the latter, because of the additional subfield subcode structure. A similar idea has been successfully exploited for Wild Goppa codes though \cite{COT17}. But in this case, the distinguisher is based on considerations of square of Goppa codes themselves, which only apply to a very restricted class of parameters.  Indeed the attack can only work for extensions of degree $m=2$ and there is no way to go beyond it, because for $m>2$ the square code fills the whole space. In our case, our distinguisher is based on squaring the {\em dual} of a Goppa code (or an alternant code) and works for {\em any} field extension degree.

However,  the fact that the dual of a Goppa code is the trace of a generalized Reed-Solomon code rather than the subfield subcode of a generalized Reed-Solomon code seems to complicate significantly the attempts to turn this distinguisher into an attack. But again, having now a much better algebraic (and rigorous) explanation of why the distinguisher works, together with new algebraic results about products involved in the square of the dual of the Goppa code gives a much better understanding of the square code structure. This is clearly desirable and needed if we want to mount a key recovery attack based on these square code considerations. The hope is that this will ultimately lead to being able to attack McEliece schemes based on very high Goppa codes. As explained earlier, this will still not threaten the codes used in the aforementioned NIST competition, but this would break the 20 years old signature scheme \cite{CFS01} that is based on very high rate Goppa codes.

\bigskip
\noindent{\bf Acknowledgements}.
This work was supported in part by the ANR CBCRYPT project, grant ANR-17-CE39-0007  and by the ANR BARRACUDA project, grant ANR-21-CE39-0009, both of the French
Agence Nationale de la Recherche. 

\newcommand{\etalchar}[1]{$^{#1}$}

\end{document}